\documentclass[a4paper, amsfonts, amssymb, amsmath, reprint, showkeys, twoside,superscriptaddress, onecolumn, nofootinbib]{revtex4-2}
\usepackage{graphicx} 
\usepackage{amsmath}
\usepackage{amssymb}
\usepackage{amsthm}
\usepackage{mathtools}
\usepackage{minted}
\usepackage{algorithm2e}
\usepackage{physics}
\usepackage{comment}
\usepackage[colorlinks=true, linkcolor=blue, citecolor=blue, urlcolor=blue]{hyperref}
\usepackage[capitalize,nameinlink]{cleveref}
\usepackage{overpic}
\usepackage{orcidlink}
\SetAlCapSkip{10pt}

\newcommand{\BCH}{\text{BCH}}
\newcommand{\bigO}{\mathcal{O}}

\newtheorem{theorem}{Theorem}

\newtheorem{corollary}[theorem]{Corollary}

\begin{document}
\title{Practical Estimation of Trotter Error for Hamiltonian Simulation}
\author{William Maxwell~\orcidlink{0009-0005-8603-2955}}
\email{william.maxwell@xanadu.ai}
\affiliation{Xanadu, Toronto, ON, MSG 2C8, Canada}
\author{Pablo A. M. Casares~\orcidlink{0000-0001-5500-9115}}
\affiliation{Xanadu, Toronto, ON, MSG 2C8, Canada}
\author{Robert A. Lang~\orcidlink{0000-0002-4345-3566}}
\affiliation{Xanadu, Toronto, ON, MSG 2C8, Canada}
\author{Stepan Fomichev~\orcidlink{0000-0002-1622-9382}}
\affiliation{Xanadu, Toronto, ON, MSG 2C8, Canada}
\author{Juan Miguel Arrazola~\orcidlink{0000-0002-0619-9650}}
\affiliation{Xanadu, Toronto, ON, MSG 2C8, Canada}
\author{Soran Jahangiri~\orcidlink{0000-0002-9988-8841}}
\affiliation{Xanadu, Toronto, ON, MSG 2C8, Canada}
\author{Ali Asadi~\orcidlink{0009-0001-9058-9052}}
\affiliation{Xanadu, Toronto, ON, MSG 2C8, Canada}
\author{Luis Alfredo Nunez Meneses~\orcidlink{0009-0003-6124-1860}}
\affiliation{Xanadu, Toronto, ON, MSG 2C8, Canada}
\author{Thomas Germain~\orcidlink{0009-0007-0289-1686}}
\affiliation{Xanadu, Toronto, ON, MSG 2C8, Canada}
\author{Danial Motlagh~\orcidlink{0009-0003-7655-4341}}
\email{danial.motlagh@xanadu.ai}
\affiliation{Xanadu, Toronto, ON, MSG 2C8, Canada}

\begin{abstract}
Trotter product formulas are a leading approach for Hamiltonian simulation on quantum computers, yet their practical performance has remained difficult to assess due to the challenge of accurately estimating the Trotter error.
In this work, we develop new theoretical results, algorithms, and software tools that advance the state-of-the-art in Trotter error estimation by orders of magnitude in both scale and accuracy.
On the theoretical side, we prove that in the asymptotic limit the error of a product formula depends on the \emph{diagonal} elements of the Baker-Campbell-Hausdorff (BCH) error operator in the eigenbasis of the Hamiltonian, rather than its full spectral norm---yielding an improved scaling for Hamiltonian simulation using product formulas.
On the algorithmic side, we introduce a compact representation of the BCH expansion that reduces the number of commutators from $\mathcal{O}(n^3)$ to $\mathcal{O}(n)$ for second-order, and from $\mathcal{O}(n^5)$ to $\mathcal{O}(n^2)$ for fourth-order formulas on $n$ fragments, complemented by an importance sampling scheme to further reduce the computational cost. We provide implementations of these techniques in software and demonstrate their power on two applications: (i) X-ray absorption spectroscopy of an electronic Hamiltonian (Li$_4$Mn$_2$O) at up to 56 qubits using tensor networks; and (ii) vibronic dynamics of naphthalene at over 100 qubits using ML-MCTDH, where we find that naive analytical bounds overestimate the required number of Trotter steps by nearly five orders of magnitude.
Our framework enables, for the first time, the accurate estimation of Trotter error at practically relevant system sizes, providing a foundation for fair algorithmic comparisons and rational design of product formulas.
\end{abstract}

\maketitle

\section{Introduction}

Simulation of quantum systems via time evolution serves as one of the most compelling use cases of quantum computers, with quantum signal processing (QSP)~\cite{low2017optimal,Low2019hamiltonian, motlagh2024generalized, berry2024doubling} and Trotter product formulas~\cite{trotter1959product, suzuki1991general} serving as the main algorithmic approaches. While QSP provides an asymptotically optimal algorithm for Hamiltonian simulation, Trotter product formulas require fewer qubits and are numerically found to be efficient in practice despite their suboptimal scaling; particularly in regimes of short time-evolution and less stringent accuracy requirements~\cite{childs2018toward, childs2021theory}. Despite this, they face the difficulty of accurately estimating their runtime a priori. However, estimating the Trotter error remains a central challenge in the development of product formulas, as the error determines the number of Trotter steps required by the algorithm, and hence its resource requirements. This makes accurate estimation of Trotter error crucial for designing efficient simulation algorithms and for elucidating the regimes where Trotter-based techniques outperform QSP-based methods.\\

Conventional approaches for computing Trotter error suffer from significant limitations in both scale and accuracy~\cite{childs2021theory, berry2007efficient}. Analytical methods typically rely on norm-based upper bounds derived from operator inequalities, which often overestimate the true runtime by orders of magnitude and therefore fail to provide useful guidance for practical implementations~\cite{berry2007efficient, wiebe2010higher, tran2020destructive}. On the other hand, exact numerical methods, such as directly computing the difference between the true evolution operator and its Trotterized approximation, rapidly become computationally infeasible as system size grows. These challenges create a substantial gap between theoretical error guarantees and practical error estimation for realistic simulation tasks.\\

In this work, we develop new theory, algorithms, and software tools that significantly surpass prior approaches to Trotter error estimation in both scale and accuracy. We begin by showing how Trotter error is best viewed through the lens of evolution under an effective Hamiltonian, $H' = H + \mathcal{E}$, as opposed to the original Hamiltonian $H$. We also prove a new theoretical result in \cref{sec:PT} showing that, in the asymptotic limit, the error of a product formula depends on the \emph{diagonal} elements of $\mathcal{E}$ in the eigenbasis of the Hamiltonian, rather than its full spectral norm, yielding an improved scaling for Hamiltonian simulation using product formulas. Next, in \cref{sec:BCH_expansion} we provide theory and software for efficiently computing the effective error $\mathcal{E}$ under a generic product formula. To do this, we introduce a compact representation that asymptotically reduces the number of commutators in the Baker-Campbell-Hausdorff (BCH) expansion. In particular, this compact representation reduces the number of commutators from $\mathcal{O}(n^3)$ to $\mathcal{O}(n)$ for second-order, and from $\mathcal{O}(n^5)$ to $\mathcal{O}(n^2)$ for fourth-order Trotter formula on $n$ fragments. We then complement this by an importance sampling scheme to further reduce the computational cost. Lastly, in \cref{sec:applications} we show how our software module, implemented in PennyLane~\cite{bergholm2018pennylane}, can be combined with state-of-the-art classical simulation techniques such tensor networks and ML-MCTDH~\cite{wang2015multilayer} to push limits of accurate Trotter error estimation far beyond what was previously possible. Specifically, we tightly estimate Trotter error for X-ray absorption spectroscopy of an electronic Hamiltonian at up to 56 qubits using tensor-networks and for vibronic dynamics at over 100 qubits using ML-MCTDH.

\section{Background}

In this section we provide an overview of the existing methods for estimating Trotter error. These methods can be categorized based on the trade-off between their predictive accuracy and computational cost. Analytical error bounds can be computed with relatively low effort, but they are pessimistic and typically overestimate resources by several orders of magnitude for practical applications~\cite{mehendale2025estimating}. Exact simulation methods, on the other hand, can provide precise Trotter error, but their computational scaling limits their applicability to small system sizes. Methods based on perturbation theory \cite{mehendale2025estimating} provide a compromise between accuracy and computational cost, and their extension and formalization provided in \cref{sec:PT} forms the basis of the approach in much of this work. We review all these methods and briefly explain their advantages and limitations.\\

Analytical approaches upper bound the Trotter error via repeated applications of the triangle inequality. At the most extreme end of the spectrum, one can upper bound the error in time evolution under a Hamiltonian $H = \sum_{j=1}^n H_j$ via the Trotter Suzuki product formula $S_{2k}(t)$ of order $2k$, using only the number of terms $n$ and the norm of the largest term $\Gamma = \text{max}_j ||H_j||$ \cite{berry2007efficient},
\begin{equation}
    \left\| e^{-iHt} - [S_{2k}(-it/r)]^r \right\| \leq \frac{5\left(2 \cdot 5^{k-1} n \cdot t \cdot \Gamma\right)^{2k+1}}{r^{2k}},
\end{equation}
or sum of norms of individual fragments~\cite{casas2026error}. This approach severely overestimates the true error and hence the number of Trotter steps required for a given accuracy. Tighter analytical bounds can be obtained at the expense of increased computational cost by instead looking at the norms of the nested commutators arising from the BCH expansion \cite{childs2021theory, casas2026error}. For example, the error from a first-order Trotter formula can be bounded via the following

\begin{equation}
    \left\| e^{-iHt} - [S_{1}(-it/r)]^r \right\| \leq \frac{t^2}{2r} \sum_{j=1}^L  \sum_{k=j+1}^L \parallel \left[H_k, H_j \right] \parallel.
\end{equation}

Expressing the error this way allows for taking into account cancellations when two terms $H_k$ and $H_j$ commute, or when their commutator can be analytically found. However, this bound is still applying a triangle inequality to bound the norm of the sum as a ``sum of the norms", and further applying another triangle inequality to bound the error from taking $r$ Trotter steps as $r$ times the error from a single Trotter step. Hence, while this approach can offer a tighter bound than the former analytical bound, it still typically overestimates the true runtime of the algorithm by orders of magnitude for practical problems.\\

At the other extreme end of the spectrum, Trotter errors can be exactly computed by constructing the full matrix form of the exact unitary $e^{-iHt}$ and the Trotter product $\left( \prod_{j} e^{-i H_j \Delta t} \right)^n$ and then computing the error as the spectral norm of their difference, or even more relevant errors such as the state and observable dependent error. However, due to the exponentially scaling dimension of the space these approaches are strictly limited to very small systems of typically less than 20 qubits, far smaller than practically relevant system sizes. \\

For eigenvalue estimation tasks, heuristic methods based on perturbation theory have been proposed \cite{mehendale2025estimating} that provide a more balanced trade-off between accuracy and computational cost. In these methods, the Trotterized circuit is treated as the evolution under an effective Hamiltonian  $H' = H + \mathcal{E}$ given by the Baker-Campbell-Hausdorff (BCH) formula. Then, given an approximate eigenstate of the Hamiltonian $\ket{\Psi}$, the error in the energy coming from a $d^{th}$ order product formula is estimated as

\begin{equation}
     \Delta E \approx \Delta t^d\,\bra{\Psi} \mathcal{E} \ket{\Psi} .
\end{equation}

In the next section, we mathematically formalize this previously heuristic approach and extend it to the case of energy estimations involving multiple eigenstates such as spectroscopy and dynamical observables.

\section{Methodology and Algorithms}

In this section we layout our methodology for estimating the number of Trotter steps required for simulating time-evolution under a Hamiltonian $H = \sum_{j=1}^n H_j$ using a $d^{th}$ order product formula
\begin{equation}
    S_d(\Delta t) = \prod_{s=1}^N e^{-i\Delta t \,a_s H_{j_s}},
\end{equation}
such that the error in the simulation is below some required accuracy $\epsilon$,
\begin{equation}
   \left\| e^{-itH} - S_d(\Delta t)^{t/\Delta t} \right\| < \epsilon,
\end{equation}
where $j_s \in \{1,2,…,m\}$ is an index sequence dictating which term $H_j$ is being applied at stage $s$. However, note that the spectral norm is a much stricter accuracy requirement than what is needed in practice, as it is the maximum possible error over all possible initial state and observable pairs. In practice, we often care about the resulting error in the value of an observable $O$ given an initial state $\ket{\psi}$,
\begin{equation}
   \left| \langle \psi(t) | O  | \psi(t) \rangle - \langle \tilde\psi(t) | O | \tilde\psi(t) \rangle \right| < \epsilon,
\end{equation}
where $\ket{\psi(t)} = e^{-iHt}\ket{\psi}$ and $| \tilde\psi(t) \rangle = S_d(\Delta t)^{t/\Delta t}\,\ket{\psi}$.\\

However, regardless of the metric we choose for accuracy, in order to find the number of Trotter steps for a $d^{th}$ order product formula, we estimate the error $\tilde\epsilon$ arising from a given step size $\tilde\Delta t$ and set the simulation step size as $\displaystyle \Delta t = \tilde\Delta t \cdot \left(\epsilon/\tilde \epsilon\,\right)^{1/d} $. Hence, estimating $\tilde\epsilon$ becomes the central task in determining the runtime of simulation algorithms based on product formulas.\\

\subsection{Perturbation theory}\label{sec:PT}

Trotter error is best viewed through the lens of the error arising from evolution under an effective Hamiltonian as opposed to the original Hamiltonian of interest,
\begin{equation} \label{eq:eff_propagator}
    S_d(\Delta t)^{t/\Delta t} = e^{-itH'},
\end{equation}
for $H' =H + (\Delta t)^d\cdot\mathcal{E}$, where $\mathcal{E}$ is given by the BCH expansion as described in \cref{sec:BCH_expansion}. Choosing $\Delta t = t^{-1/d}$, we get
\begin{equation}
    S_d(\Delta t)^{t/\Delta t} = e^{-i(tH +\mathcal{E})}.
\end{equation}
We now introduce one of our key results: in the asymptotic limit, this effective Hamiltonian becomes diagonal in the basis of the original Hamiltonian.
\begin{theorem}\label{thm:operator_limit}
    Given a Hamiltonian $H =\sum_k \lambda_k \, \ket{\lambda_k}\bra{\lambda_k}$ and a bounded Hermitian operator $\mathcal{E}$, we have that 
\[
\lim_{t\to \infty} e^{i(tH+ \mathcal{E})} = e^{i(tH+ D(\mathcal{E}))},
\]
where $D(\mathcal{E})= \sum_n \bra{\lambda_n}\mathcal{E}\ket{\lambda_n}\, \ket{\lambda_n}\bra{\lambda_n} $.\\
\end{theorem}
We prove \cref{thm:operator_limit} in Appendix \ref{app:PT_proofs}. The above result leads to the following corollary.
\begin{corollary}\label{cor:PF_limit}
    Given a Hamiltonian $H =\sum_k \lambda_k \, \ket{\lambda_k}\bra{\lambda_k}$ and a product formula $S_d(\Delta t)$, such that
\[
S_d(\Delta t) = e^{i\Delta t\left(H + \Delta t^d\cdot\mathcal{E}\right)},
\]
let $U(t)= S_d\left(t^{-1/d}\right)^{t^{(1+1/d)}}$, then   
\[
\lim_{t\to \infty} U(t) = e^{i(tH+ D(\mathcal{E}))},
\]
where $D(\mathcal{E})= \sum_n \bra{\lambda_n}\mathcal{E}\ket{\lambda_n}\, \ket{\lambda_n}\bra{\lambda_n} $.\\
\end{corollary}
This leads to the implication that, in the asymptotic limit, the scaling of product formulas depend on the diagonal elements of $\mathcal{E}$ (in the eigenbasis of $H$) as opposed to its spectral norm as previously thought. More specifically, the number steps for a $d^{th}$ order product formula scales as $\mathcal{O}\left(\frac{t^{(1+1/d)}\cdot \|D(\mathcal{E})\|^{1/d}}{\epsilon^{1/d}}\right)$ as opposed to the previously established scaling of $\mathcal{O}\left(\frac{t^{(1+1/d)}\cdot \|\mathcal{E}\|^{1/d}}{\epsilon^{1/d}}\right)$.
\begin{theorem}\label{thm:error_limit}
    Given a Hamiltonian $H =\sum_k \lambda_k \, \ket{\lambda_k}\bra{\lambda_k}$ and a product formula $S_d(\Delta t)$, such that
\[
S_d(\Delta t) = e^{i\Delta t\left(H + \Delta t^d\cdot\mathcal{E}\right)},
\]
let $U(t)= S_d\left(\Delta t\right)^{t/\Delta t}$ for $\Delta t = \delta \cdot t^{-1/d}$, then   
\[
\lim_{t\to \infty} \left\|U(t) - e^{itH}\right\|= \left\|e^{i \delta^d \cdot D(\mathcal{E})} - \mathbb{I}\right\|\leq \delta^d \cdot \| D(\mathcal{E}) \|,
\]
and for any eigenstate $\ket{\lambda_n}$
\[
\lim_{t\to \infty} \left\|\left(U(t) - e^{itH}\right)\ket{\lambda_n}\right\|= \left|e^{i \delta^d \cdot \bra{\lambda_n}\mathcal{E}\ket{\lambda_n}} - 1\right|\leq \delta^d \cdot | \bra{\lambda_n}\mathcal{E}\ket{\lambda_n} |.
\]
Furthermore, given an observable $O$ we have
\[
\lim_{t\to \infty} \left\|\tilde O(t) - O(t)\right\|= \left\|e^{i \delta^d \cdot D(\mathcal{E})}\, O \, e^{-i \delta^d \cdot D(\mathcal{E})} - O\right\|\leq \delta^d \cdot \| \,[D(\mathcal{E}), O] \,\|,
\]
for $O(t) = e^{itH}Oe^{-itH}$ and $\tilde O(t) = U(t)\,O\,U(t)^\dagger$.\\
\end{theorem}

This result not only establishes an improved scaling for Hamiltonian simulation using product formulas, but serves as the foundation for our approach to estimating Trotter error. While the above results are derived for $t\to \infty$, they still provide reasonable approximations for finite $t$ due to the sublinear scaling of the number of Trotter steps in $\epsilon$. For instance, approximating $\tilde \epsilon$ up to a $100\%$ relative error results in at most a $2^{1/d}$ relative error in the estimation of the number of Trotter steps. Hence, our strategy becomes clear: estimate $\langle \tilde\lambda_n | \mathcal{E} | \tilde\lambda_n \rangle$ for classically prepared approximate eigenstates $| \tilde\lambda_n \rangle$, and use \cref{thm:error_limit} to estimate the error of the product formula. To do this we need a way to compute $\mathcal{E}$, which we describe in the following subsection.

\subsection{Computing the effective Hamiltonian}\label{sec:BCH_expansion}
The main ingredient in our Trotter error estimation workflow is the Baker-Campbell-Hausdorff (BCH) formula~\cite{campbell1896law,baker1905alternants,hausdorff1906symbolische}. The BCH formula provides a solution $C$ to the equation $e^A e^B = e^C$, where $A, B$, and $C$ are elements of a Lie algebra. The solution $C$ is given as an infinite sum of nested commutators.
For a Hamiltonian expressed as a sum of fragments $H = \sum_{i=1}^n H_i$, we represent an order-$k$ commutator by $[H_{i_1}, \dots, H_{i_k}]$, where $i_j \in \{1, \dots, n\}$. This notation hides the bracketing structure of the commutator and will be used to denote an arbitrary order-$k$ commutator in cases where the bracketing structure is not important.
Computing the Trotter error of $H$ requires evaluating the commutators of the BCH expansion up to some fixed order $k$. Asymptotically, there are $\bigO(n^k)$ such commutators which corresponds to the dimension of the subspace of the free Lie algebra spanned by the order $k$ commutators.\\

The main challenge of computing Trotter error is the $\bigO(n^k)$ scaling for the number of commutators.
Computing the effective Hamiltonian requires evaluating the sum of the commutators, which requires computing expensive fragment products that appear when expanding the commutator by $[H_i, H_j] = H_iH_j - H_jH_i$. For perturbation error, we have to sum over expectation values of commutators with respect to an eigenstate of the form $\bra{\psi}[H_i, H_j] \ket{\psi}$, resulting in expensive matrix-vector products like $H_i \ket{\psi}$. Even for second order formulas the $\bigO(n^3)$ scaling can be a major bottleneck.\\

 The goal of this section is to compute the BCH expansion in a compact representation containing asymptotically fewer commutators. We derive expressions for $k=3$ and $k=5$ as sums of $\bigO(n)$ and $\bigO(n^2)$ commutators respectively. This is a transformative improvement compared to the $\bigO(n^3)$ and $\bigO(n^5)$ bounds implied by the dimensions of their subspaces. An algorithm based on the ideas in this section has been implemented in PennyLane's Trotter error module.\\

At a high level, our compact representation is obtained by recursively computing the symmetric BCH expansion in the Hall \cite{casas2009efficient} basis, and algebraically combining terms at each step. The end result is a heavily reduced number of commutators of the form
\begin{equation}\label{eq:compact_comm}
\left[ \sum_{i=1}^n \alpha_{1,i} H_i, \dots, \sum_{i=1}^n \alpha_{k,i} H_i \right],
\end{equation}
where the individual Lie algebra generators have been replaced with linear combinations of generators. In practice, computing sums of fragments is much cheaper than products ($H_iH_j$ for the effective Hamiltonian or $H_i \ket{\psi}$ for perturbation error), so this tradeoff leads to large performance gains.

\subsubsection{The symmetric BCH formula}\label{sec:symm_bch}
The symmetric BCH formula solves the equation $e^{B/2}e^Ae^{B/2} = e^C$ expressing $C$ as an infinite sum of nested commutators.
The advantage of the symmetric BCH formula is that all commutators of even order cancel out, which is crucial in deriving our compact representations. In practice, most product formulas of interest are symmetric (such as the Trotter-Suzuki hierarchy) and benefit from these cancellations.
The symmetric BCH expansion up to fifth order commutators is given by the following.

\begin{equation}\label{eq:symm_bch}
\begin{aligned}
\BCH (A, B) &= \log \left( e^{B/2} e^A e^{B/2} \right) \\
&= A + B - \frac{\left[\left[A,B\right],A\right]}{12} - \frac{\left[\left[A,B\right],B\right]}{24} - \frac{\left[\left[\left[A,B\right],A\right],\left[A,B\right]\right]}{360} + \frac{\left[\left[\left[A,B\right],B\right],\left[A,B\right]\right]}{480}\\
&+ \frac{\left[\left[\left[\left[A,B\right],A\right],A\right],A\right]}{720} + \frac{\left[\left[\left[\left[A,B\right],B\right],A\right],A\right]}{180} + \frac{7 \left[\left[\left[\left[A,B\right],B\right],B\right],A\right]}{1440} + \frac{7 \left[\left[\left[\left[A,B\right],B\right],B\right],B\right]}{5760}
 + \dots.\\
\end{aligned}
\end{equation}

We need to compute the symmetric BCH expansion on the full set of generators, and put them in the form given by \cref{eq:compact_comm}. To do this, we apply the symmetric BCH expansion recursively, given by

\begin{align}
\BCH( H_1, \dots, H_n) &\coloneqq \log \left(e^{H_n/2}\ldots e^{H_2/2}e^{H_1}e^{H_2/2}\ldots e^{H_n/2}\right) \\
&= \BCH\left( \BCH( \ldots\BCH(H_1, H_2), \dots,) H_n \right). \label{eq:bch_recurse}
\end{align}

By recursively applying \cref{eq:bch_recurse} and applying the bilinearity of commutators, $[A, B] + [A, C] = [A, B + C]$, we obtain our compact representations presented in \cref{sec:reduce_comms}.

\subsubsection{Recursively defined product formulas}
We extend the approach of \cref{sec:symm_bch} to recursively defined product formulas. As an example, we focus on the Trotter-Suzuki hierarchy, but in practice this approach will apply to any recursively defined symmetric product formula. Recall the Trotter-Suzuki hierarchy given by 
\begin{align}
\mathcal{U}_2(t) &= \left( \prod_{i=n}^1 e^{tH_i / 2} \prod_{i=1}^n e^{tH_i / 2} \right)\\
\mathcal{U}_{2k}(t) &= \mathcal{U}_{2k - 2} \left( u_kt \right)^2
\mathcal{U}_{2k -2} \left( \left( 1 - 4u_k \right) t \right) 
\mathcal{U}_{2k-2} \left( u_kt \right)^2,
\end{align}
where $u_k = 1 / \left( 4 - 4^{1 / (k+1)} \right)$.
The naive approach to computing $\log \left(\mathcal{U}_{2k}(t) \right)$ is to simply expand and apply \cref{eq:bch_recurse} as
\begin{equation}
    \log \left( \mathcal{U}_{2k}(t) \right) = \log \left( e^{t\alpha_1 H_{i_1}} \dots e^{t \alpha_m H_{i_m}} \right)\\
    = \BCH \left(t \alpha_1 H_{i_1}, \dots, t \alpha_m H_{i_m} \right),
\end{equation}
but this requires computing $\BCH$ on $m = 5^{k-1} n$ elements of the Lie algebra, which is infeasible even for small values of $k$. Instead, we work in the free Lie algebra generated by the product formulas $\mathcal{V}_1 \coloneqq \mathcal{U}_{2k-2}(u_kt)^2$ and $\mathcal{V}_2 \coloneqq \mathcal{U}_{2k-2}((1-4u_k)t)$ and compute
\begin{equation}\label{eq:bch_pf}
    \log \left( \mathcal{U}_{2k}(t) \right) = \BCH \left( 2 \mathcal{V}_1, \mathcal{V}_2 \right),
\end{equation}
then recursively compute $\log ( \mathcal{V}_1)$ and $\log(\mathcal{V}_2)$, substituting $\mathcal{V}_1$ and $\mathcal{V}_2$ for their expansions in the expression for $\log \left( \mathcal{U}_{2k}(t) \right)$. The base case is $\mathcal{U}_2$ which is computed by applying \cref{eq:bch_recurse}.\\

The BCH implementation in PennyLane leverages recursively defined product formulas for performance with caching. A product formula in the form $\mathcal{P}_0 \mathcal{P}_1 \mathcal{P}_0 \mathcal{P}_1 \mathcal{P}_0$ is computed by calling $\BCH(2 \mathcal{P}_0, 2 \mathcal{P}_1, \mathcal{P}_0)$, and we can cache the results from the first $\log ( \mathcal{P}_0 )$ computation to be used later for the second computation.

\subsubsection{Reducing the number of commutators}\label{sec:reduce_comms}
In this section we present our compact representations for the third and fifth order commutators of the BCH expansion.
Let $Y_{k,n}$ denote the symmetric BCH expansion on $n$ elements restricted to just the order $k$ commutators. For $H = \sum_{i=1}^n H_i$, we express the linear order term in the effective Hamiltonian as
\begin{equation}
    Y_{1,n} = H_n + Y_{1,n-1} = \sum_{i=1}^n H_i.
\end{equation}
Summing over $k$ we obtain the full symmetric BCH expansion
\begin{equation}
    \text{BCH}(H_1, \dots, H_n) = \sum_{k=1}^\infty Y_{k, n}.
\end{equation}
In order to count the number of order $k$ commutators in the symmetric BCH expansion we need to obtain an expression for $Y_{k, n}$.\\ 

We obtain explicit expressions for $Y_{3, n}$ and $Y_{5, n}$. Similar expressions can be found for higher orders following the same method. Our starting point is the symmetric BCH expansion in the Hall basis given in \cref{eq:symm_bch}; we recursively apply \cref{eq:bch_recurse} and apply the bilinearity of the commutator $[A, B] + [A, C] = [A, B + C]$. The final expression is a sum of commutators in the form of \cref{eq:compact_comm}. For $k=3$ we obtain the expression
\begin{equation}
\begin{aligned}
Y_{3,n} &= Y_{3,n-1} - \frac{\left[ \left[ Y_{1,n-1}, H_n \right], Y_{1, n-1} \right] }{12} - \frac{\left[ \left[ Y_{1,n-1}, H_n \right], H_n \right] }{24} \\
&= Y_{3, n-1} - \frac{\left[ \left[ \sum_{i=1}^{n-1} H_i, H_n \right], \sum_{i=1}^{n-1} H_i \right] }{12} - \frac{\left[ \left[ \sum_{i=1}^{n-1} H_i, H_n \right], H_n \right] }{24}.
\end{aligned}
\end{equation}
Each fragment adds two commutators to the expression of $Y_{3, n}$, making the total number of commutators $\bigO(n)$, a significant reduction from $\bigO(n^3)$.
The same technique yields the following expression for $k=5$,
\begin{equation}\label{eq:Y_5n}
\begin{aligned}
    Y_{5,n} &= Y_{5,n-1}  - \frac{\left[\left[Y_{3,n-1}, H_n\right],H_n\right]}{24} 
    - \frac{\left[\left[Y_{3,n-1},H_n\right],\sum_{i=1}^{n-1} H_i \right]}{12}
    - \frac{\left[\left[\sum_{i=1}^{n-1} H_i,H_n\right],Y_{3,n-1} \right]}{12}\\
    &+ \frac{\left[\left[\left[\left[\sum_{i=1}^{n-1} H_i,H_n\right],\frac{1}{4}\sum_{i=1}^{n-1} H_i + H_n\right],\sum_{i=1}^{n-1} H_i\right],\sum_{i=1}^{n-1} H_i\right]}{180}\\
    &+ \frac{7 \left[\left[\left[\left[\sum_{i=1}^{n-1} H_i,H_n\right],H_n\right],H_n\right],\sum_{i=1}^{n-1} H_i + \frac{1}{4}H_n\right]}{1440}\\
    &- \frac{\left[\left[\left[\sum_{i=1}^{n-1} H_i,H_n\right],\frac{1}{3}\sum_{i=1}^{n-1} H_i - \frac{1}{4}H_n\right],\left[\sum_{i=1}^{n-1} H_i,H_n\right]\right]}{120}.
\end{aligned}
\end{equation}
In $Y_{5,n}$ there are commutators containing $Y_{3, n}$, leading to $\bigO(n^2)$ commutators in total --- a massive reduction from $\bigO(n^5)$. The implementation in PennyLane constructs $Y_{3, n}$ exactly and empirically matches the asymptotic growth of $Y_{5, n}$.\\

\noindent This reduction was crucial for obtaining the results presented in \cref{sec:applications} for two reasons:
\begin{enumerate}
    \item Our algorithm significantly decreases the time needed to compute the BCH expansion of a product formula.
    \item Our algorithm significantly decreases the number of commutators needed to evaluate to the Trotter error.
\end{enumerate}
\cref{fig:fourthorder} shows these reductions for a fourth-order Trotter formula.

\begin{figure*}
    \centering
    \begin{overpic}[width=0.45\linewidth, percent]{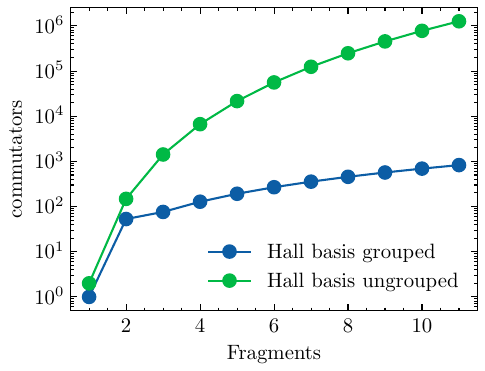}
        \put(1.,73){(a)}
    \end{overpic}
    \quad
    \begin{overpic}[width=0.45\linewidth, percent]{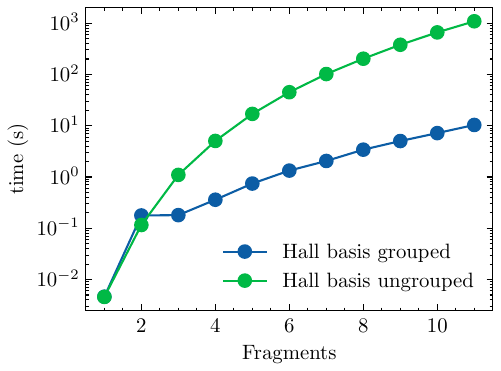}
        \put(1,73){(b)}
    \end{overpic}
    \caption{Computing the BCH expansion of fourth order Trotter-Suzuki using the \textsf{bch\_expansion} function provided by PennyLane's Trotter error module. Comparing the algorithm with and without the commutator reduction method of \cref{sec:reduce_comms}. Figure (a) plots the number of commutators returned by BCH with and without grouping and Figure (b) plots the time needed to compute the BCH expansion. }
    \label{fig:fourthorder}
\end{figure*}

\subsection{Importance sampling}
In practice the fragments of the Hamiltonian $H = \sum_{i=1}^n H_i$ are not uniformly weighted. Fragments with higher norms contribute more to the Trotter error than ones with smaller norms. As a result, the commutators of the BCH expansion are not of uniform importance and we can obtain a good approximation of the Trotter error by only considering the commutators of high importance. We thus allow strategies that only sample the top $k$ most important commutators and compute the error estimate from them. This importance sampling approach and the commutator reduction technique in \cref{sec:BCH_expansion} complement each other: sampling the top $k$ commutators in the form of \cref{eq:compact_comm} is equivalent to sampling the top $k$ linear combinations of commutators from the original BCH expansion, leading to a faster convergence, as shown in \cref{fig:convergence_commutators}.\\

The importance sampling heuristic is implemented in the following way. Assume we are given weights for the fragments of $H$ as a function $w \colon \{H_1, \dots, H_n\} \rightarrow \mathbb{R}^+$. A typical choice for $w$ would be the spectral norm. We extend $w$ to a function on the free Lie algebra generated by $H_1, \dots, H_n$. First, we extend $w$ linearly such that $w \left( \sum_{i=1}^n \alpha_i H_i \right) = \sum_{i=1}^n \alpha_i w(H_i)$. Then we extend $w$ to commutators by
\begin{equation}
w \left( \left[ \sum_{i=1}^n \alpha_{1,i} H_i, \dots, \sum_{i=1}^n \alpha_{k,i} H_i \right]_T \right) = 2^{k-1} \prod_{i=1}^k \sum_{j=1}^n \alpha_{j,i} w(H_j).
\end{equation}
This heuristic is motivated by the fact that when $w$ is a norm we can use the inequality $\|[A, B]\| \leq 2\|A\|\|B\|$. In \cref{sec:applications} we demonstrate the usefulness of the heuristic in practice.

\begin{figure*}[t!]
    \centering
    \begin{overpic}[width=0.45\linewidth, percent]{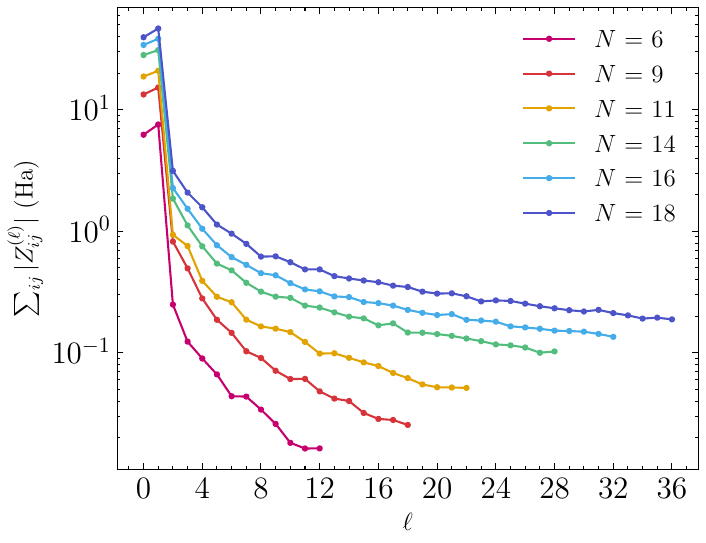}
        \put(1,73){(a)}
    \end{overpic}
    \hspace{0.05\linewidth}
    \begin{overpic}[width=0.45\linewidth, percent]{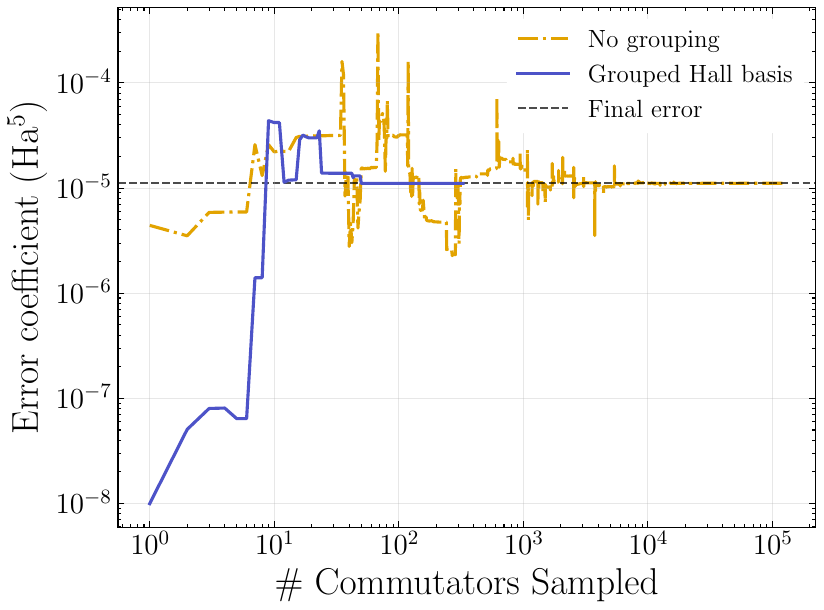}
        \put(1,73){(b)}
    \end{overpic}
    \caption{(a) Frobenius norms of the CDF coefficient matrices $Z^{(\ell)}$ for the Li$_4$Mn$_2$O Hamiltonian. The rapid decay underpins the effectiveness of importance sampling: commutators involving high-norm fragments dominate the Trotter error. (b) Importance sampling convergence for  the case of $N = 6$ with a fourth-order Suzuki formula. We show the estimated error as a function of the number of commutators sampled in order of decreasing importance with and without the grouping methodology introduced in \cref{sec:reduce_comms}. This demonstrates the complementary nature of the grouping method with importance sampling.}
    \label{fig:convergence_commutators}
\end{figure*}
\section{Applications}\label{sec:applications}

The value of a Trotter error estimation framework is ultimately measured by the scale of the practical problems it can tackle. In this section we apply our methods pipeline to two physically motivated applications: (i) X-ray absorption spectroscopy of an electronic Hamiltonian, which utilizes commutator grouping, importance sampling, and tensor network evaluation of perturbation error in concert; and (ii) state population dynamics under a vibronic Hamiltonian, where we combine our techniques with state-of-the-art classical packages for dynamics to push the limits of exact (state and observable)-dependent Trotter error up to systems requiring 100+ qubits.

\subsection{Spectroscopy for electronic Hamiltonians}\label{sec:electronic}

We first consider the X-ray absorption spectroscopy (XAS) of Li$_4$Mn$_2$O, a lithium-excess transition-metal oxide cluster that serves as a model for battery cathode materials~\cite{fomichev2024simulating,fomichev2025fast}. This is the same system studied in~\cite{casares2026theory}, where the Trotter error estimates produced by our methods fed directly into fault-tolerant resource estimates and contributed to an overall $\times 4.5$ reduction in Toffoli gate count. Prior perturbation-theory-based Trotter error analyses for electronic Hamiltonians have relied on operator-norm bounds extrapolated from small active spaces~\cite{fomichev2025fast}, or were limited to exact diagonalization at $\lesssim 20$ qubits. Here, we evaluate the actual leading-order perturbative error directly at the target system size, for active spaces of up to $N = 28$ spatial orbitals (56 qubits) -- the largest system for which such estimates have been obtained.\\

In spectroscopy tasks the Trotter error manifests as energy shifts relative to the true spectrum~\cite{malpathak2025trotter} due to evolution under an effective Hamiltonian arising from the product formula. We decompose the electronic Hamiltonian using Compressed Double Factorization (CDF)~\cite{cohn2021quantum}, which approximates the two-electron integrals as
\begin{align}\label{eq:CDF}
   (pq|rs) \approx \sum_{\ell=1}^L \sum_{k,l=1}^N U^{(\ell)}_{pk} U^{(\ell)}_{qk} Z^{(\ell)}_{kl} U^{(\ell)}_{rl} U^{(\ell)}_{sl},
\end{align}
yielding $L$ Hamiltonian fragments whose Frobenius norms $\|Z^{(\ell)}\|_F$ decay rapidly with $\ell$ (\cref{fig:convergence_commutators}). Given a $p$th-order product formula with effective Hamiltonian $H_{\mathrm{eff}} = H + \Delta t^p\,\mathcal{E} + O(\Delta t^{p+1})$, first-order perturbation theory gives the eigenvalue shift
\begin{equation}\label{eq:perturbation_eigenvalues}
    E'_l = E_l + \Delta t^p \bra{E_l}\mathcal{E}\ket{E_l} + O(\Delta t^{p+1}),
\end{equation}
where $\ket{E_l}$ are eigenstates of $H$ with eigenvalues $E_l$. The computational task is to evaluate $\bra{E_l}\mathcal{E}\ket{E_l}$, where $\mathcal{E}$ is a linear combination of nested commutators of the Hamiltonian fragments returned by our BCH expansion module.\\

Let us first consider the naive approach: expand the BCH series for a $p$th-order product formula on $L$ fragments, then evaluate every nested commutator expectation value individually. For the second-order Trotter formula, the leading error operator $\mathcal{E}$ contains $O(L^3)$ third-order commutators, and for the fourth-order Trotter formula it contains $O(L^5)$ fifth-order commutators. With CDF the number of fragments scales as $L \sim N$, so at the largest active space ($N = 28$), the fourth-order formula requires evaluating on the order of $28^5 \approx 1.7 \times 10^7$ commutator expectation values with respect to approximate eigenstate living in a Hilbert space of dimension $2^{56}$. This is completely intractable using previous approaches. We now trace how three innovations combine to make this evaluation tractable at unprecedented scale.\\

\paragraph{Commutator grouping.}
The grouping algorithm of \cref{sec:BCH_expansion} replaces this brute-force expansion with a far more compact one by collecting commutators of individual fragments into commutators of \emph{partial sums} of fragments. This lets us trade expensive operator--state multiplications with cheap operator--operator additions, so the grouped expression has strictly fewer terms that require taking an expectation value with respect to approximate eigenstates. For the second-order formula, the number of commutator-state products drops from $O(L^3)$ to $O(L)$; for the fourth-order formula, from $O(L^5)$ to $O(L^2)$ -- a cubic reduction. At $N = 28$ this takes $\sim\!10^7$ evaluations down to $\sim\!10^3$ for the fourth-order formula, a reduction by four orders of magnitude. Without grouping, applying a fourth-order formula to any system with more than $\sim\!10$ fragments is out of reach; with it, systems with tens of fragments become tractable.\\

\paragraph{Importance sampling.}
Even after grouping, evaluating all $O(L^2)$ grouped commutators is still expensive because each requires taking an expectation value in an exponentially large space. Importance sampling eliminates most of this remaining work by exploiting the hierarchical norm structure of the CDF: because the Frobenius norms $\|Z^{(\ell)}\|_F$ decay rapidly \cref{fig:convergence_commutators}(a), the commutator contributions follow a heavy-tailed distribution in which a small fraction of terms dominates the total error. We assign each grouped commutator an importance derived from the fragments' norms and evaluate only the top-$k$ contributors. \Cref{fig:convergence_commutators}(b) shows that for the $N = 6$ fourth-order Suzuki case, the top $\sim\!50$ grouped commutators (out of several hundred) capture the full error to high precision. Importance sampling also compounds with grouping: each sampled grouped commutator represents a linear combination of many raw BCH terms, so evaluating $k$ grouped commutators effectively captures far more than $k$ terms from the original expansion.\\

\paragraph{Tensor-network evaluation.}
The final bottleneck is evaluating the surviving commutator expectation values. At $N = 28$ the Hilbert space has dimension $2^{56}$, far beyond exact diagonalization. We address this through an interface which allows any data type to be plugged into the pipeline that implements basic operator--operator arithmetic and operator--state multiplication. For electronic Hamiltonians we provide a backend built on the block2 tensor-network library~\cite{zhai2023block2}, representing approximate eigenstates as matrix product states (MPS) obtained via DMRG and nested commutators as matrix product operators (MPO).  To compute any matrix element of a commutator, we unroll the commutator into a linear combination of products of operators. For instance,
\begin{equation}
    \langle[A, B]\rangle = \langle AB\rangle - \langle BA\rangle,
\end{equation}
and each expectation value is evaluated by contracting MPOs and MPSs. Starting with bond dimensions $\chi_{MPO}$ and $\chi_{MPS}$ for the MPOs and MPSs, we approximate the resulting MPS with a bond dimension
\begin{equation}
    \chi = \sqrt{\chi_{MPO}^2 + \chi_{MPS}^2}.
\end{equation}
We find that bond dimensions ranging from $\chi = 25$ to $\chi=100$ yields well-converged error estimates, with $\chi = 25$ estimates closely resembling those of $\chi = 100$.\\

This newly found ability to accurately estimate Trotter error directly at the target scale, rather than using pessimistic bounds or extrapolations enables two key contributions: a more fair comparison between simulation algorithms based on product formulas and those based on qubitization, and the enablement of rational design of product formulas by allowing a more accurate comparison between different product formula constructions.

\subsection{Dynamics for vibronic Hamiltonians}\label{sec:vibronic}

In this section we discuss how our methods for obtaining the effective Hamiltonian under a Trotterized product formula can be combined with state-of-the-art classical simulation methods such as multilayer multiconfiguration time-dependent Hartree (ML-MCTDH)~\cite{wang2003multilayer, wang2015multilayer} and density matrix renormalization group (DMRG) \cite{schollwock2011density}, to tightly estimate Trotter error for dynamical problems at scales far beyond what was previously possible. Specifically, we first show how DMRG can be used in conjunction with our methods to improve bounds on Trotter runtime estimates by several orders of magnitude as compared to bounds derived based on triangle inequality. Next we show how to further improve the bounds via direct simulation of (state and observable)-dependent Trotter error for system sizes beyond 100 qubits, a task completely infeasible previously.\\

As an example, we focus on the task of computing the time-dependent electronic state populations under a vibronic Hamiltonian describing the coupled dynamics of electronic and nuclear degrees of freedom of a molecule \cite{kouppel1984multimode, domcke2004conical}. We consider a second-order Trotter product formula based on the fragmentation scheme of Ref.~\cite{motlagh2024quantum} for a vibronic Hamiltonian of the form
\begin{align}
    H &= \mathbb{I}_{\mathrm{el}} \otimes \frac{1}{2}\sum_{r}^{M} \omega_r (P_r^2 + Q_r^2) \,\, + \,\,\sum_{i,j}^N\sum_{r}^{M} \,\ket{i}\!\bra{j}_{\mathrm{el}} \otimes \left( \alpha_r^{(i,j)} \,Q_r + \beta_r^{(i,j)} \,Q_r^2 \right),
\end{align}
where $Q_r$ and $P_r$ are the position and momentum operators for the $r^{\mathrm{th}}$ vibrational mode with frequency $\omega_r$. In this setting, the observables of interest correspond to the time-dependent population of each electronic state $\ket{i}$, and take the form
\begin{align}\label{eq:state_pop_obs}
    O_i = \ket{i}\!\bra{i}_{\mathrm{el}} \otimes \mathbb{I}_{\mathrm{nuc}}.
\end{align}
As a concrete system, we analyze a vibronic model of naphthalene with 2 electronic state and all 48 vibrational modes. Naphthalene is the prototypical polycyclic aromatic hydrocarbon, and its vibronically-coupled $S_1$ and $S_2$ $\pi\pi^*$ states drive ultrafast, non-adiabatic internal conversion ~\cite{montero2010coherent}. Such photophysics are of relevance to organic optoelectronics, such as OLED emitters and organic semiconductors~\cite{liu2015high}. We systematically assess how the Trotter error scales with system size by constructing a hierarchy of reduced-dimensionality models by ranking the 48 vibrational normal modes according to their coupling strengths. Each mode is discretized into 32 gridpoints, meaning an $M$-mode reduced-dimensionality model corresponds to a $(5M +1)$-qubit system. We use our methods to obtain the effective Hamiltonian $H'= H + \Delta t^d\, \mathcal{E}$ for the workflows described below.

\begin{figure}[t]
    \centering
    \includegraphics[width=0.9\linewidth]{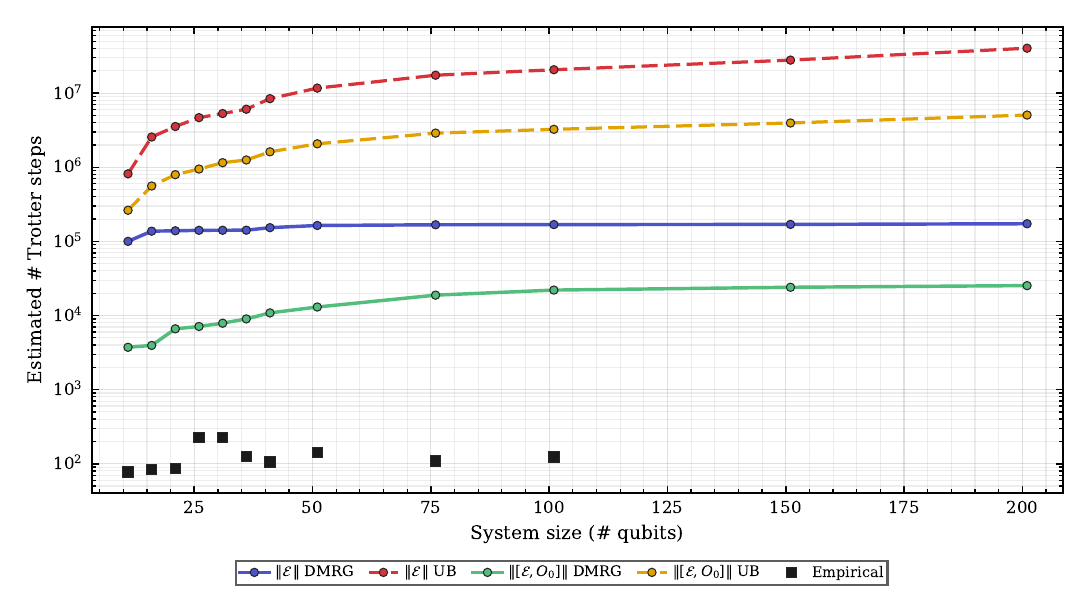}
    \caption{%
        Number of second-order Trotter steps $r$ required to achieve
        a maximum state-population error of $\leq 0.01$ over $100\,\mathrm{fs}$
        of evolution time, as a function of qubit count ($5M + 1$ qubits for $M$ vibrational modes),
        for the reduced-dimensionality naphthalene vibronic models.
        Five estimation methods are compared:
        the naive triangle-inequality upper bounds (UB) on $\|\mathcal{E}\|$ and $\|[\mathcal{E}, O_0]\|$,
        the DMRG-computed spectral norms of $\mathcal{E}$ and $[\mathcal{E}, O_0]$,
        and the empirical estimate from direct ML-MCTDH simulations. The DMRG spectral norms provide bounds orders of magnitude tighter
        than the naive upper bounds, while the direct dynamical estimate
        yields the most realistic assessment of the required number of Trotter steps.}
    \label{fig:vibronic_steps_vs_qubits}
\end{figure}
\subsubsection{Tighter bounds via DMRG}

Given an accuracy requirement $\epsilon$ for the simulation, upper bounds on the number of Trotter steps are typically derived based on the norm of $\mathcal{E}$. For observable-independent bounds this corresponds to finding a $\Delta t$ such that
\begin{equation}
    t\,\Delta t^d\,\|\mathcal{E}\|\leq \epsilon,
\end{equation}
and for observable-dependent bounds it corresponds to
\begin{equation}
    t\,\Delta t^d\,\|[\mathcal{E}, O]\|\leq \epsilon.
\end{equation}
However, since $\mathcal{E}$ lives in an exponentially large Hilbert space where direct diagonalization is inaccessible, its norm is typically bounded via repeated applications of the triangle inequality $\|\mathcal{E}\|\leq \sum_i \|\mathcal{E}_i\|$ for terms $\mathcal{E}_i$ coming from the BCH expansion, whose norms are known analytically. This leads to a gross overestimation of $\|\mathcal{E}\|$ and thus the number of Trotter steps. To avoid this, we instead use DMRG to directly estimate $\|\mathcal{E}\|$, and find an improvement in the estimated runtime by two orders of magnitude for the system studied here, as seen in \cref{fig:vibronic_steps_vs_qubits}.\\

These improved bounds, while orders of magnitude tighter than naive triangle inequalities, remain state-independent and thus capture only the worst-case scenario corresponding to unphysical states of very high energy. In the next section we describe how we can further reduce runtime estimates by additional orders of magnitude when taking into account the initial state in the dynamical problem of interest.

\subsubsection{Empirical estimates via ML-MCTDH}

Direct estimation of (state and observable)-dependent Trotter error is the most accurate method for estimating resource requirements for a particular simulation problem of interest. More specifically, given an initial state $\ket{\psi_0}$, we wish to find the maximal step size $\Delta t$ such that
\begin{equation}\label{eq:state_dep_err}
    \sup_{\tau\leq t} \left|  \bra{\psi_0}\, e^{iH\tau}\, O\, e^{-iH\tau}\,\ket{\psi_0}
    -
    \bra{\psi_0}\, e^{iH'\tau}\, O\, e^{-iH'\tau}\ket{\psi_0}\right|\leq \epsilon.
\end{equation}
This is typically inaccessible for systems beyond 20 qubits using state-vector simulations. Instead we employ the state-of-the-art classical method for simulating dynamics, ML-MCTDH, to empirically probe (state and observable)-dependent Trotter error directly for systems with over 100 qubits. We use the \texttt{pyTTN} package~\cite{lindoy2025pyttn} implementation of ML-MCTDH to time propagating $\ket{\psi_0}$ under both the original Hamiltonian $H$ and the effective $H'$ produced by our methods, and compare the difference in the resulting state populations. This not only allows us to obtain significantly smaller estimates on the runtime of the Trotter algorithm but also leads to the insight that the number of required Trotter steps does not necessarily grow monotonically with system size, as suggested by norm-based bounds. We plot the estimated runtime using ML-MCTDH against bounds derived via DMRG and naive triangle inequality in \cref{fig:vibronic_steps_vs_qubits}. At the end, we find that the naive triangle inequality overestimates the number of Trotter steps by almost five orders of magnitude at 20 modes with 32 grid points per mode (100+ qubits).\\

\section{Conclusion}

We have presented a comprehensive framework for the practical estimation of Trotter error, advancing the state-of-the-art in both theoretical understanding and computational capability. Our contributions span three interconnected directions: new theory that sharpens the asymptotic scaling of product formulas, algorithms that drastically reduce the cost of computing the BCH error operator, and a software implementation in PennyLane that interfaces with state-of-the-art classical simulation methods to push Trotter error estimation to unprecedented system sizes.\\

On the theoretical side, we proved that in the asymptotic limit, the error of a product formula is governed by the diagonal elements of the BCH error operator $\mathcal{E}$ in the eigenbasis of the Hamiltonian, rather than its full spectral norm. This yields an improved scaling for Hamiltonian simulation, replacing the dependence on $\|\mathcal{E}\|$ with $\|D(\mathcal{E})\|$, and provides a rigorous foundation for perturbation-theory-based error estimates that were previously only heuristic.\\

On the algorithmic side, our compact representation of the BCH expansion reduces the number of commutators from $\bigO(n^3)$ to $\bigO(n)$ for second-order and from $\bigO(n^5)$ to $\bigO(n^2)$ for fourth-order product formulas on $n$ fragments. Complemented by our importance sampling scheme, these reductions make it feasible to compute tight Trotter error for product formulas on Hamiltonians with tens of fragments, a regime previously out of reach for fourth-order formulas.\\

On the applications side, we demonstrated the usefulness of our framework on two physically motivated problems. For X-ray absorption spectroscopy of Li$_4$Mn$_2$O, we obtained tight perturbative error estimates at up to 56 qubits using tensor networks---the largest system for which such estimates have been produced. For vibronic dynamics of naphthalene, we combined our software with ML-MCTDH to directly probe state- and observable-dependent Trotter error at over 100 qubits, a task completely infeasible with previous state-vector simulation techniques, revealing that naive analytical bounds overestimate the required number of Trotter steps by nearly five orders of magnitude.\\

These results have immediate practical implications. First, they enable fair comparisons between Trotter-based and QSP-based simulation algorithms by providing realistic---rather than pessimistic---resource estimates for product formulas. Second, they open the door to the rational design of product formulas, where different constructions and orderings can be evaluated and optimized based on accurate error predictions rather than loose upper bounds. Third, the modularity of our software ensures that our framework can be readily extended to new Hamiltonian types and classical simulation backends as they become available.\\

\bibliography{main}
\appendix

\section{Perturbation theory proofs}\label{app:PT_proofs}

\begin{proof}[Proof of \cref{thm:operator_limit}]
    Consider the Dyson series expansion of $e^{i(tH+ \mathcal{E})}= e^{itH}\cdot U(1, 0; t)$, where
    \begin{align}\label{eq:bigboy}
        U(1, 0; t) =& \Biggl(\,\mathbb{I} + i\int_{0}^1 dt_1e^{-itHt_1}\mathcal{E}e^{itHt_1}+ \left(i\right)^2\int_{0}^1 dt_1 \int_{0}^{t_1} \, dt_2 e^{-itHt_1}\mathcal{E}e^{itHt_1}e^{-itHt_2}\mathcal{E}e^{itHt_2} \\
        &\cdots+ \left(i\right)^k\int_{0}^1 dt_1\int_{0}^{t_1} dt_2 \cdots \int_{0}^{t_{k-1}} dt_k\,e^{-itHt_1}\mathcal{E}e^{itHt_1} \cdots e^{-itHt_k}\mathcal{E}e^{itHt_k} +\cdots\Biggr).
        \end{align}
        We define $U_k(t)$ for convenience as
        \begin{equation}
            U_k(t) = \left(i\right)^k\int_{0}^1 dt_1\int_{0}^{t_1} dt_2 \cdots \int_{0}^{t_{k-1}} dt_k\,e^{-itHt_1}\mathcal{E}e^{itHt_1} \cdots e^{-itHt_k}\mathcal{E}e^{itHt_k},
        \end{equation}
        and $U_0(t) = \mathbb{I}$. Thus 
        \begin{equation}
            e^{i(tH+ \mathcal{E})} = e^{itH}\,\sum_{k=0}^\infty U_k(t) .
        \end{equation}
        We wish to prove
        \begin{align}
            \lim_{t\to \infty} e^{i(tH+ \mathcal{E})} = e^{i(tH+ D(\mathcal{E}))},
        \end{align}
        where $D(\mathcal{E})= \sum_n \bra{\lambda_n}\mathcal{E}\ket{\lambda_n}\, \ket{\lambda_n}\bra{\lambda_n} = \sum_n \mathcal{E}_{nn}\, \ket{\lambda_n}\bra{\lambda_n}$. Since $[H, D(\mathcal{E})]=0$, $e^{i(tH+ D(\mathcal{E}))}= e^{itH}e^{iD(\mathcal{E})}$. Therefore, it suffices to prove
        \begin{equation}
            \lim_{t\to \infty} \sum_{k=0}^\infty U_k(t) = e^{iD(\mathcal{E})},
        \end{equation}
        or equivalently for all eigenstates $\ket{\lambda_n}$ of $H$
        \begin{equation}
            \lim_{t\to \infty} \sum_{k=0}^\infty U_k(t) \ket{\lambda_n}= e^{i\mathcal{E}_{nn}}\ket{\lambda_n}.
        \end{equation}
        Let $\Pi = \ket{\lambda_n}\bra{\lambda_n}$, then for each term in the sum, $U_k(t)$, we insert $k$ resolutions of identity $\mathbb{I} = \left(\mathbb{I} - \Pi \right) + \Pi $ after the first $k$ applications of $\mathcal{E}$
        \begin{align}
           \sum_{k=1}^\infty & U_k(t)\ket{\lambda_n} = \sum_{k=1}^\infty\sum_{\vec{J}\in \{0,1\}^{k-1}} \Phi_k(\vec{J}, t)\ket{\lambda_n},
        \end{align}
        For
        \begin{align}
            \Phi_k(\vec{J}, t) := \left(i\right)^k\int_{0}^1 dt_1\int_{0}^{t_1} dt_2 \cdots \int_{0}^{t_{k-1}} dt_k\,e^{-itHt_1}\left(\mathbb{I} -\Pi\right)^{(1-J_1)} \Pi^{J_1}\mathcal{E}e^{itHt_1} \cdots e^{-itHt_k}\left(\mathbb{I} -\Pi\right)^{(1-J_{k})} \Pi^{J_{k}}\mathcal{E}e^{itHt_k}.
        \end{align}
        Separating the case of $\vec{J}=\vec{1}$ from $\vec{J}\neq\vec{1}$ (where $\vec{1}$ is the all ones vector), we get
        \begin{align}
           \sum_{k=1}^\infty & U_k(t)\ket{\lambda_n} = \sum_{k=1}^\infty \Phi_k(\vec{1}, t)\ket{\lambda_n} + \sum_{k=1}^\infty\sum_{\substack{\vec{J}\in \{0,1\}^{k-1}\\ \vec{J}\neq\vec{1}}} \Phi_k(\vec{J}, t)\ket{\lambda_n}.
        \end{align}
        It's easy to check
        \begin{equation}
            \Phi_k(\vec{1}, t)\ket{\lambda_n} =  \frac{ \left(i\,\mathcal{E}_{nn}\right)^k}{k!} \ket{\lambda_n}.
        \end{equation}
        Thus
        \begin{align}
           \sum_{k=0}^\infty U_k(t)\ket{\lambda_n} &= \sum_{k=0}^\infty\frac{ \left(i\,\mathcal{E}_{nn}\right)^k}{k!}\,\ket{\lambda_n} + \sum_{k=1}^\infty\sum_{\substack{\vec{J}\in \{0,1\}^{k-1}\\ \vec{J}\neq\vec{1}}} \Phi_k(\vec{J}, t)\ket{\lambda_n},\\
           &= e^{i \,\mathcal{E}_{nn}}\ket{\lambda_n} + \sum_{k=1}^\infty\sum_{\substack{\vec{J}\in \{0,1\}^{k-1}\\ \vec{J}\neq\vec{1}}} \Phi_k(\vec{J}, t)\ket{\lambda_n}.
        \end{align}
        Hence all that remains to show is
        \begin{equation}
            \lim_{t\to\infty}\left\| \sum_{k=1}^\infty\sum_{\substack{\vec{J}\in \{0,1\}^{k-1}\\ \vec{J}\neq\vec{1}}} \Phi_k(\vec{J}, t)\ket{\lambda_n} \right\| = 0.
        \end{equation}
        We now bound $\left\|\Phi_k(\vec{J}, t)\ket{\lambda_n}\right\|$ for the case of $\Vec{J}\neq\Vec{1}$. Let $\Vec{J}\in \{0,1\}^{k-1}$ such that $\Vec{J}\neq\Vec{1}$. Then let $c$ be the highest index in $\Vec{J}$ such that $J_c=0$ (that is the first time we cross over to the $(\mathbb{I}-\Pi)$ subspace), we have
        \begin{align}
             \Phi_k(\vec{J}, t)\ket{\lambda_n} = \left(i\right)^k\mathcal{E}_{nn}^{k-c-1} &\int_{0}^1 dt_1e^{-itHt_1}\left(\mathbb{I} -\Pi\right)^{(1-J_1)} \Pi^{J_1}\mathcal{E}e^{itHt_1}\int_{0}^{t_1} dt_2 e^{-itHt_2}\left(\mathbb{I} -\Pi\right)^{(1-J_1)} \Pi^{J_1}\mathcal{E}e^{itHt_2} \cdots\nonumber\\ &\cdots \int_{0}^{t_{c-1}} e^{-itHt_{c}}\left(\mathbb{I} -\Pi\right)\mathcal{E}e^{itHt_{c}} \frac{t_{c}^{k-c}}{(k-c)!}\ket{\lambda_n} \,dt_{c},
        \end{align}
        thus
        \begin{align}
             \left\|\Phi_k(\vec{J}, t)\ket{\lambda_n}\right\| &\leq \|\mathcal{E} \|^{k-1}\int_{0}^1 dt_1\int_{0}^{t_1} dt_2 \cdots \left\|\int_{0}^{t_{c-1}} e^{-itHt_{c}}\left(\mathbb{I} -\Pi\right)\mathcal{E}e^{itHt_{c}} \frac{t_{c}^{k-c}}{(k-c)!}\ket{\lambda_n} \,dt_{c}\right\|.
        \end{align}
        Let us examine $\int_{0}^{t_{c-1}} e^{-itHt_{c}}\left(\mathbb{I} -\Pi\right)\mathcal{E}e^{itHt_{c}} \frac{t_{c}^{k-c}}{(k-c)!}\ket{\lambda_n} \,dt_{c}$
        \begin{align}
             \int_{0}^{t_{c-1}} e^{-itHt_{c}}\left(\mathbb{I} -\Pi\right)\mathcal{E}e^{itHt_{c}} \frac{t_{c}^{k-c}}{(k-c)!}\ket{\lambda_n} \,dt_{c} = \sum_{m\neq n} \frac{\mathcal{E}_{mn}}{(k-c)!} \ket{\lambda_m}\int_{0}^{t_{c-1}} e^{it(\lambda_n - \lambda_m)t_{c}} \cdot t_{c}^{k-c} \,dt_{c}
        \end{align}
        It is easy to check via integration by parts that
        \begin{equation}
            \Biggl|\int_{0}^{t_{c-1}} e^{it(\lambda_n - \lambda_m)t_{c}} \cdot t_{c}^{k-c} \,dt_{c}\Biggr|\leq \frac{2\cdot t_{c-1}^{k-c}}{t\cdot|\lambda_n - \lambda_m|}
        \end{equation}
        Let $\gamma = \min_{m\neq n} |\lambda_n - \lambda_m|$, then
        \begin{align}
            \left\|\int_{0}^{t_{c-1}} e^{-itHt_{c}}\left(\mathbb{I} -\Pi\right)\mathcal{E}e^{itHt_{c}} \frac{t_{c}^{k-c}}{(k-c)!}\ket{\lambda_n} \,dt_{c}\right\| &\leq \frac{2\cdot t_{c}^{k-c}}{t\cdot \gamma \cdot (k-c)!}\cdot \left\|\sum_{m\neq n} \mathcal{E}_{mn} \ket{\lambda_m}\right\|,\\
            &\leq \frac{2\cdot t_{c}^{k-c}}{t\cdot\gamma \cdot (k-c)!}\cdot \left\|\mathcal{E}\right\|.
        \end{align}
        Thus
        \begin{align}
             \left\|\Phi_k(\vec{J}, t)\ket{\lambda_n}\right\| &\leq  \frac{1}{t}\frac{2\cdot \left\|\mathcal{E}\right\|^k}{\gamma} \int_{0}^1 dt_1\int_{0}^{t_1} dt_2 \cdots \int_{0}^{t_{c-2}} \frac{t_{c-1}^{k-c}}{(k-c)!} \,dt_{c-1},\\
             &= \frac{1}{t} \cdot \frac{2\cdot \left\|\mathcal{E}\right\|^k}{\gamma \cdot (k-1)!}
        \end{align}
        Then
        \begin{align}
            \left\| \sum_{k=1}^\infty\sum_{\substack{\vec{J}\in \{0,1\}^{k-1}\\ \vec{J}\neq\vec{1}}} \Phi_k(\vec{J}, t)\ket{\lambda_n} \right\| &\leq  \sum_{k=1}^\infty2^k  \left\|\Phi_k(\vec{J}, t)\ket{\lambda_n} \right\|\\
            &\leq  \sum_{k=1}^\infty2^k \frac{1}{t} \cdot \frac{2\cdot \left\|\mathcal{E}\right\|^k}{\gamma \cdot (k-1)!}\\
            &\leq \frac{1}{t} \frac{4\cdot \left\|\mathcal{E}\right\|}{\gamma} \sum_{k=1}^\infty \frac{(2\left\|\mathcal{E}\right\|)^{k-1}}{(k-1)!}\\
            &= \frac{1}{t} \frac{4\cdot \left\|\mathcal{E}\right\|\cdot e^{2\left\|\mathcal{E}\right\|}}{\gamma}
        \end{align}
        Since $\frac{4\cdot \left\|\mathcal{E}\right\|\cdot e^{2\left\|\mathcal{E}\right\|}}{\gamma}$ is just a constant, then
        \begin{equation}
            \lim_{t\to\infty}\left\| \sum_{k=1}^\infty\sum_{\substack{\vec{J}\in \{0,1\}^{k-1}\\ \vec{J}\neq\vec{1}}} \Phi_k(\vec{J}, t)\ket{\lambda_n} \right\| = 0.
        \end{equation}
        Which completes our proof. However, it's worth noting that while the constant $\frac{4\cdot \left\|\mathcal{E}\right\|\cdot e^{2\left\|\mathcal{E}\right\|}}{\gamma}$ may imply logarithmic convergence to this limit, in practice this convergence happens much quicker due to 2 reasons: $\frac{4\cdot \left\|\mathcal{E}\right\|\cdot e^{2\left\|\mathcal{E}\right\|}}{\gamma}$ is a very loose upper bound on the residual since we applied many layers of triangle inequalities to simplify the proof. Secondly, when setting $\Delta t = \delta \cdot t^{-1/d}$, $\left\|\mathcal{E}\right\|$ is replaced with $\delta^d\left\|\mathcal{E}\right\|$ in the expression for the upper bound on the residual.\\
\end{proof}
\bigskip 
\begin{proof}[Proof of Corollary~\ref{cor:PF_limit}]
Setting $\Delta t = t^{-1/d}$ in the assumed relation $S_d(\Delta t) = e^{i\Delta t(H + \Delta t^d \cdot \mathcal{E})}$, we obtain
\begin{equation}
S_d(t^{-1/d}) = e^{i\, t^{-1/d}\left(H + t^{-1} \cdot \mathcal{E}\right)}.
\end{equation}
Then
\begin{align}
U(t) &= \left[S_d(t^{-1/d})\right]^{t^{(1+1/d)}} \nonumber\\
&= e^{i(tH + \mathcal{E})}.
\end{align}
The result then follows by direct application of \cref{thm:operator_limit}.
\end{proof}
\bigskip

\begin{proof}[Proof of \cref{thm:error_limit}]
We proceed in three parts corresponding to the three claims.

\paragraph{Part 1: Spectral norm bound.}
Setting $\Delta t = \delta \cdot t^{-1/d}$, we compute
\begin{align}
S_d(\delta \cdot t^{-1/d}) &= e^{i\,\delta\, t^{-1/d}(H + (\delta\, t^{-1/d})^d \cdot \mathcal{E})} \nonumber\\
&= e^{i\,\delta\, t^{-1/d}(H + \delta^d\, t^{-1} \cdot \mathcal{E})}.
\end{align}
The number of Trotter steps is
\begin{equation}
\frac{t}{\Delta t} = \frac{t}{\delta \cdot t^{-1/d}} = \frac{t^{(1+1/d)}}{\delta}.
\end{equation}
Therefore,
\begin{align}
U(t) &= \left[S_d(\delta \cdot t^{-1/d})\right]^{t^{(1+1/d)}/\delta} \nonumber\\
&= e^{i\,(t^{(1+1/d)}/\delta)\cdot \delta\, t^{-1/d}\cdot (H + \delta^d t^{-1}\mathcal{E})} \nonumber\\
&= e^{i\, t^{(1+1/d)}\cdot t^{-1/d}\cdot(H + \delta^d t^{-1}\mathcal{E})} \nonumber\\
&= e^{i\,t\,(H + \delta^d t^{-1}\mathcal{E})} \nonumber\\
&= e^{i(tH + \delta^d \mathcal{E})}.
\end{align}
Applying Theorem~1 with the bounded Hermitian operator $\delta^d \mathcal{E}$:
\begin{equation}\label{eq:asymp_U}
\lim_{t\to\infty} e^{i(tH + \delta^d \mathcal{E})} = e^{i(tH + \delta^d D(\mathcal{E}))}.
\end{equation}
Since $D(\mathcal{E})$ is diagonal in the eigenbasis of $H$, the operators $H$ and $D(\mathcal{E})$ commute, so
\begin{equation}
e^{i(tH + \delta^d D(\mathcal{E}))} = e^{itH}\, e^{i\delta^d D(\mathcal{E})}.
\end{equation}
Therefore,
\begin{align}
\lim_{t\to\infty}\left\|U(t) - e^{itH}\right\| &= \left\|e^{itH}\,e^{i\delta^d D(\mathcal{E})} - e^{itH}\right\| \nonumber\\
&= \left\|e^{itH}\left(e^{i\delta^d D(\mathcal{E})} - \mathbb{I}\right)\right\| \nonumber\\
&= \left\|e^{i\delta^d D(\mathcal{E})} - \mathbb{I}\right\|,
\end{align}
where the last step uses the unitary invariance of the spectral norm. For the upper bound, we use the standard inequality $\|e^{iA} - \mathbb{I}\| \leq \|A\|$ valid for any Hermitian operator $A$ (which follows from $|e^{i\theta} - 1| \leq |\theta|$ applied to each eigenvalue). Setting $A = \delta^d D(\mathcal{E})$:
\begin{equation}
\left\|e^{i\delta^d D(\mathcal{E})} - \mathbb{I}\right\| \leq \delta^d \|D(\mathcal{E})\|.
\end{equation}

\paragraph{Part 2: Eigenstate bound.}
Acting on an eigenstate $\ket{\lambda_n}$ and using Eq.~\eqref{eq:asymp_U}:
\begin{align}
\lim_{t\to\infty}\left\|\left(U(t) - e^{itH}\right)\ket{\lambda_n}\right\| &= \left\|\left(e^{itH} e^{i\delta^d D(\mathcal{E})} - e^{itH}\right)\ket{\lambda_n}\right\|, \nonumber\\
&= \left\|\left(e^{i\delta^d D(\mathcal{E})} - \mathbb{I}\right)\ket{\lambda_n}\right\|,\nonumber\\
&=|e^{i\delta^d \bra{\lambda_n}\mathcal{E}\ket{\lambda_n}} - 1|,\nonumber\\
&\leq \delta^d\cdot| \bra{\lambda_n}\mathcal{E}\ket{\lambda_n}|
\end{align}

\paragraph{Part 3: Observable bound.}
The Heisenberg-picture observable under exact evolution is $O(t) = e^{itH} O\, e^{-itH}$, and under the Trotterized evolution is $\tilde{O}(t) = U(t)\, O\, U(t)^\dagger$. Using Eq.~\eqref{eq:asymp_U} and its adjoint:
\begin{align}
\lim_{t\to\infty} \tilde{O}(t) &= e^{i(tH+\delta^d D(\mathcal{E}))}\, O\, e^{-i(tH+\delta^d D(\mathcal{E}))} \nonumber\\
&= e^{itH}\,e^{i\delta^d D(\mathcal{E})}\, O\, e^{-i\delta^d D(\mathcal{E})}\, e^{-itH}.
\end{align}
Therefore,
\begin{align}
\lim_{t\to\infty}\left\|\tilde{O}(t) - O(t)\right\| &= \left\|e^{itH}\,e^{i\delta^d D(\mathcal{E})}\,O\,e^{-i\delta^d D(\mathcal{E})}\,e^{-itH} - e^{itH}\,O\,e^{-itH}\right\| \nonumber\\
&= \left\|e^{i\delta^d D(\mathcal{E})}\,O\,e^{-i\delta^d D(\mathcal{E})} - O\right\|,
\end{align}
where we again used unitary invariance of the spectral norm to remove $e^{\pm itH}$. For the upper bound, we use the integral representation
\begin{equation}
e^{iA}\,O\,e^{-iA} - O = i\int_0^1 e^{isA}[A, O]\,e^{-isA}\,ds,
\end{equation}
which gives
\begin{align}
\left\|e^{iA}\,O\,e^{-iA} - O\right\| &\leq \int_0^1 \left\|e^{isA}[A,O]\,e^{-isA}\right\| ds \nonumber\\
&= \int_0^1 \|[A, O]\|\, ds \nonumber\\
&= \|[A, O]\|.
\end{align}
Setting $A = \delta^d D(\mathcal{E})$:
\begin{equation}
\left\|e^{i\delta^d D(\mathcal{E})}\,O\,e^{-i\delta^d D(\mathcal{E})} - O\right\| \leq \|[\delta^d D(\mathcal{E}), O]\| = \delta^d\,\|[D(\mathcal{E}), O]\|.
\end{equation}
This completes the proof.
\end{proof}

\section{Computational details for \cref{sec:vibronic}}

\Cref{fig:vibronic_dynamical_error}(a) displays the time-resolved population error $\Delta P_i(t)$ for the six-mode naphthalene model at three representative Trotter step sizes $\Delta t \in \{0.01,\, 0.1,\, 1.0\}\,\mathrm{fs}$. The empirical error decreases by precisely two orders of magnitude with each tenfold reduction in $\Delta t$, consistent with second-order scaling. \Cref{fig:vibronic_dynamical_error}(b) confirms that this $\Delta t^{2}$ behaviour persists across the full reduced-dimensionality series from $M = 2$ to $M = 8$ modes. To quantify this scaling, we define the maximum population error observed over a target evolution time of $100\,\mathrm{fs}$,
\begin{align}\label{eq:max_pop_error}
    \epsilon_{\Delta t}
    =
    \sup_{\tau \leq 100\,\mathrm{fs},\; i}
    \left| \Delta P_i(\tau;\, \Delta t) \right|,
\end{align}
and fit the empirical scaling $\epsilon_{\Delta t} \propto \Delta t^{\alpha}$, where the exponent $\alpha$ is obtained as the least-squares slope of $\log\epsilon_{\Delta t}$ versus $\log\Delta t$ over the range $\Delta t \in \{0.01, 0.1, 1.0\}\,\mathrm{fs}$. All models exhibit $\alpha \approx 2$ to high precision, confirming the expected behaviour of a second-order product formula. The robustness of this power-law scaling allows us to exploit a single simulation at an accessible step size $\Delta t_{\mathrm{obs}}$ to predict the error at any target step size. Specifically, from the observed $\alpha = 2$ scaling, the ratio of errors at two different step sizes satisfies
\begin{align}\label{eq:error_ratio}
    \frac{\epsilon_{\Delta t_{\mathrm{obs}}}}{\epsilon_{\Delta t_{\mathrm{targ}}}} = \frac{\Delta t_{\mathrm{obs}}^{2}}{\Delta t_{\mathrm{targ}}^{2}},
\end{align}
from which the number of Trotter steps $r_{\mathrm{targ}} = t/\Delta t_{\mathrm{targ}}$ required to achieve a target accuracy $\epsilon_{\mathrm{targ}}$ is immediately obtained as
\begin{align}\label{eq:r_targ}
    r_{\mathrm{targ}} = \frac{t}{\Delta t_{\mathrm{obs}}} \sqrt{\frac{\epsilon_{\Delta t_{\mathrm{obs}}}}{\epsilon_{\mathrm{targ}}}}\,.
\end{align}

\begin{figure}
    \centering
    \includegraphics[width=0.95\linewidth]{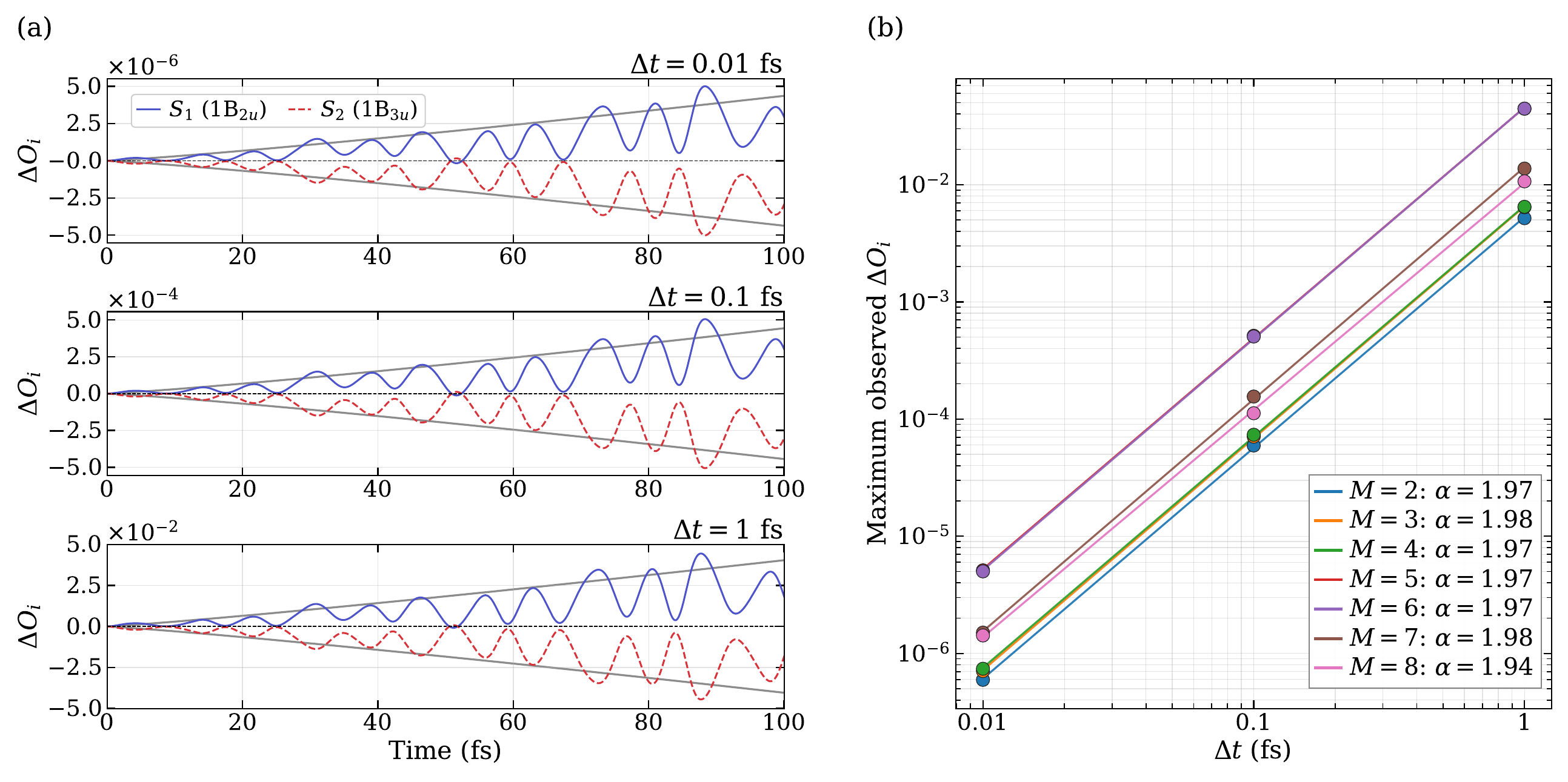}
    \caption{Empirical Trotter error for the two-state naphthalene vibronic model. (a)~Time-resolved state-population Trotter error $\Delta O_i(t)$ for three Trotter step sizes $\Delta t \in \{0.01,\,0.1,\,1\}\,\mathrm{fs}$ in the $M=6$ reduced-dimensionality model. (b)~Scaling of the maximum population error $\epsilon_{\Delta t}$ with Trotter step size $\Delta t$ for reduced-dimensionality models ranging from $M = 2$ to $M = 8$ modes. The empirical power-law exponent $\alpha$ (slope of the log--log fit) is reported for each model, demonstrating extremely robust $\Delta t^{2}$ scaling consistent with the second-order product formula across three orders of magnitude in step size.}
    \label{fig:vibronic_dynamical_error}
\end{figure}

In \cref{fig:vibronic_steps_vs_qubits} we report the estimated number of second-order Trotter steps required to achieve a maximum state-population error $\Delta O_i \leq 0.01$ over $100\,\mathrm{fs}$ of evolution across the reduced-dimensionality naphthalene series. Each vibrational mode is discretized onto 32 real-space grid points, requiring 5 qubits per mode, with a single additional qubit encoding the two dimensional electronic space, for a total of $5M + 1$ qubits in the $M$-mode model. Using the leading-order bounds and demanding $\Delta O_i \leq \epsilon$ yields the required number of steps for $d=2$ as 
\begin{align}
 r & \;\geq\;  \sqrt{\frac{t^{3}\,\|\mathcal{E}\|}{\epsilon}}\,  \quad \qquad \quad \textrm{(observable-independent)}   \\ 
 r & \;\geq\;  \sqrt{\frac{t^{3}\,\|[\mathcal{E},\, O_i]\|}{\epsilon}}\, \quad \quad \textrm{(observable-dependent).}
\end{align}
The DMRG-computed spectral norms $\|\mathcal{E}\|$ and $\|[\mathcal{E}, O_i]\|$ provide bounds that are orders of magnitude tighter than the naive triangle-inequality upper bounds (denoted ``UB'' in \cref{fig:vibronic_steps_vs_qubits}). Notably, the DMRG-computed spectral norms exhibit a clear saturation with increasing $M$: because modes are added in order of decreasing vibronic importance, the first several modes dominate the norms of both $\mathcal{E}$ and $[\mathcal{E}, O_i]$, and subsequent modes contribute negligibly. This saturation is a direct consequence of our mode-importance ranking. The strongly coupled modes that drive the non-adiabatic dynamics are selected first, so the error operators are effectively converged well before all 48 modes are included.\\

The empirical dynamical error extracted from the ML-MCTDH simulations reaches its maximum at $M = 5$, corresponding to a 26-qubit system, where approximately $220$ Trotter steps suffice to achieve $\Delta O_i \leq 0.01$, in contrast to the more than $10^6$ steps estimated by the naive triangle-inequality bounds. Interestingly, for $M > 5$ the maximum population error decreases, suggesting a partial cancellation of error contributions from the additional weakly coupled modes in the time-resolved population dynamics. A related effect was reported in Ref.~\citenum{bay2026quantum}, where significant cancellations were observed in Trotter errors on eigenvalue differences. Since unitary dynamics is dictated by energy differences rather than absolute energies, correlated Trotter errors that shift the spectrum approximately uniformly preserve the gap structure and hence suppress dynamical errors well beyond what norm-based bounds predict.

\end{document}